\documentclass[11pt]{article}

\usepackage{amssymb,amscd, amsmath}
\usepackage{amsthm}
\usepackage{latexsym}
\usepackage{euscript}
\usepackage{hyperref}

\usepackage{geometry}
\usepackage{layout}
\usepackage[utf8]{inputenc}
\usepackage[T1]{fontenc}

\usepackage{mathrsfs,multirow}
\usepackage{ae}
\usepackage{url}
\usepackage[boxed]{algorithm2e}

\usepackage{verbatim}
\usepackage{mathenv}

\newtheorem{definition}{Definition}
\newtheorem{theorem}[definition]{Theorem}
\newtheorem{proposition}[definition]{Proposition}
\newtheorem{lemma}[definition]{Lemma}

\def\prob#1{{\mathbb P}\left(#1\right)}

\newcommand{\A}{{\mathcal A}}
\newcommand{\E}{{\mathcal E}}

\newcommand{\rank}{{\rm rank}}

\title{RankSign : an efficient signature algorithm based on the rank metric}

\author{Philippe Gaborit\thanks{Universit\'e de Limoges, XLIM-DMI,
    123, Av. Albert Thomas, 87060 Limoges Cedex,
    France. {\tt gaborit,schrek,ruatta@unilim.fr}} \and Olivier
  Ruatta$^*$ \and Julien Schrek$^*$ \and Gilles Z\'emor\thanks{Universit\'e de
    Bordeaux, Institut de Mathématiques, UMR 5251, {\tt zemor@math.u-bordeaux.fr}}}

\begin{document}

\maketitle

\begin{abstract}
We propose a new approach to code-based signatures that
makes use in particular of rank metric codes. When the classical approach consists in finding the unique
preimage of a syndrome through a decoding algorithm, we propose to introduce the notion
of mixed decoding of erasures and errors for building signature schemes. 
In that case the difficult problem becomes, as is the case in lattice-based cryptography,
finding a preimage of weight above the Gilbert-Varshamov bound (case where
many solutions occur) rather than finding a unique preimage of weight below
the Gilbert-Varshamov bound. The paper describes RankSign: a
 new signature algorithm for the rank metric 
based on a new mixed algorithm for decoding erasures and errors for
the recently introduced Low Rank Parity Check (LRPC) codes.
We explain how it is possible (depending on choices
of parameters) to obtain a full decoding algorithm which is able 
to find a preimage of reasonable rank weight for any random syndrome
with a very strong probability. We study the semantic security
of our signature algorithm and show how it is possible to reduce
the unforgeability to direct attacks on the public matrix, so that 
no information leaks through signatures. Finally, we give several examples of parameters
for our scheme, some of which  with public key of size $11,520$ bits and signature of size $1728$ bits.
Moreover the scheme can be very fast for small base fields.
\end{abstract}

{\bf Keys words: post-quantum cryptography, signature algorithm, code-based cryptography, rank metric}


\section{Introduction}

In the last few years there has been a burst of activity in post-quantum cryptography.
Interest for the field has indeed increased significantly since the recent attacks on 
the discrete logarithm problem in small characteristic {\cite{Joux}, which 
shows that finding new attacks on classical
cryptographic systems is always a possibility and that it is important
to have alternatives.

Among potential candidates for alternative cryptography, lattice-based and code-based cryptography 
are strong candidates. In this paper we consider the signature problem for code-based cryptography
and especially rank metric based cryptography.
The problem of finding an efficient signature algorithm 
has been a major challenge for code-based cryptography since its introduction in 1978
by McEliece. Signing with error-correcting codes can be achieved in different ways:
the CFS algorithm \cite{CFS} considers extreme parameters of Goppa codes to obtain
a class of codes in which a non-negligeable part of random syndromes are invertible.
This scheme has a very small signature size, however it is rather slow and the public
key is very large. Another possibility is to use the Fiat-Shamir heuristic to turn
a zero-knowledge authentication scheme (like the Stern authentication scheme \cite{Ste96})
into a signature scheme. This approach leads
to very small public keys of a few hundred bits and is rather fast, but the signature
size in itself is large (about 100,000b), so that overall no wholly satisfying scheme
is known. 

Classical code-based cryptography relies on the Hamming distance 
but it is also
possible to use another metric: the rank metric. This metric introduced 
in 1985 by Gabidulin \cite{Gab85} is very different from the Hamming
distance. The rank metric has 
received
in recent years very strong attention from the coding community
because of its relevance 
to network coding. Moreover, this metric can also be used for cryptography. Indeed
it is possible to construct rank-analogues of Reed-Solomon codes: the
Gabidulin codes. Gabidulin codes inspired early cryptosystems, like the GPT cryposystem
(\cite{GPT91}), but they turned out to be
inherently vulnerable because of the
very strong structure of the underlying codes. More recently, by considering an approach similar
to NTRU \cite{ntru}(and also MDPC codes \cite{mdpc}) 
 constructing a very efficient cryptosystem
based on weakly structured rank codes was shown to be possible~\cite{GMRZ13}.
However, in terms of signatures based on the rank metric, 
only systems that use Fiat-Shamir are presently known \cite{GSZ11}. 
Overall the main interest of rank-metric based cryptography is that the complexity of the
best known attack grows very quickly with the size of parameters:
Contrary to (Hamming) code-based or to lattice-based cryptography,
it is possible to obtain {\it a general instance of the rank decoding problem} 
with size only a few thousands bits for (say) $2^{80}$ security, 
when such parameter sizes can be obtained only with
additional structure (quasi-cyclic for instance) for code-based or lattice based cryptography.

An interesting point in code-based  cryptography is that in general the security
of the protocols relies on finding small weight vectors {\it below} the Gilbert-Varshamov bound
(the typical minimum weight of a random code).
This is noticeably different from lattice based cryptography for which it is very common for
the security of a signature algorithm \cite{GGH,MR08} to rely on the capacity to approximate a random vector
far beyond its closest lattice vector element (the Gap-CVP problem).

Traditionally, this approach was not developed for code-based cryptography since no 
decoding algorithm is known that decodes beyond the Gilbert-Varshamov bound: 
in fact this problem is somewhat marginal for
 the coding community since it implies
many possibilities for decoding, while the standard objective is 
to find the most probable codeword or a short list of most likely codewords.

 
{\bf Our contribution} 

The main contribution of this paper is the introduction of a new way of considering
code-based signatures, by introducing the idea that it is possible to invert
a random syndrome not below the Gilbert-Varshamov bound, but above it.
The approach is similar in spirit to what is done in lattice-based
cryptography. We describe a new algorithm for LRPC codes, a
recently introduced class of rank codes, the new algorithm allows in
practice to
decode both errors and (generalized) rank erasures. This new algorithm
enables us to approximate a syndrome beyond the Gilbert-Varshamov bound.
The algorithm is a unique
decoder (not a list decoder) but can give different solutions depending
on the choice of the erasure. We shall explain precisely in which conditions
one can obtain successful decoding for any given syndrome and give the
related probabilistic analysis.
Based on this error/erasure algorithm we propose a new signature
scheme --
RankSign. We give conditions for which no information leakage is possible
from real signatures obtained through our scheme. This is a
significant point
since information leaking from real signatures was the weakness
through which the NTRUSign scheme came 
to be attacked \cite{ntrusign,DuPh12,PhRe06}. Finally, we give examples of parameters:
they are rather versatile, and their size depends on a bound on the amount
of potentially leaked information. In some cases one obtains
public keys of size 11,000 bits with signatures of length 1728 bits, moreover
the scheme is rather fast.

The paper is organized as follows: Section~\ref{sec:background} recalls basic facts on
the rank metric, Section~\ref{sec:approximating}
introduces LRPC codes and describes a new mixed algorithm for decoding (generalized) erasures and errors,
and studies its behaviour, Section~\ref{sec:RankSign} shows how to use them for
cryptography, and lastly, Section 5 and 6  consider security and parameters for these schemes. The details of some proofs and attacks are also given in the appendix. 
 
\section{Background on rank metric codes and cryptography}
\label{sec:background}
\subsection{Definitions and notation}
        \textbf{Notation :}
        Let $q$ be a power of a prime $p$, $m$ an integer and let $V_n$ be a
$n$ dimensional vector space over the finite field ${\rm GF}(q^m)$.  Let $\beta =(\beta _1,\dots ,\beta _m)$ be a basis of $GF(q^m)$ over $GF(q)$.\\ Let $\mathcal{F}_i$ be the map from 
$GF(q^m)$ to $GF(q)$ where $\mathcal{F}_i(x)$ is the $i$-th coordinate of $x$ in the basis $\beta$.\\ To any $v=(v_1,\dots ,v_n)$ in $V_n$ we associate the matrix $\overline{v} \in \mathcal{M}_{m,n}(F_{q})$ in which $\overline{v}_{i,j}=\mathcal{F}_i(v_j)$. \\ The rank weight of a vector $v$ can be defined as the rank of the associated matrix $\overline{v}$. If we name this value ${\rm rank}(v)$ we can define a distance between two vectors $x,y$ through the formula ${\rm d_r}(x,y)={\rm rank}(x-y)$. 
{\it Isometry for rank metric:} in the
rank metric context, the notion of isometry differs from the Hamming
distance context:
while for Hamming distance isometries are permutation matrices, for
the rank metric isometries
are invertible $n \times n$ matrices on the base field $GF(q)$ (indeed these matrices, usually
denoted by $P$, do not change the rank of a codeword).  We refer to \cite{Loi06} for more details on codes for the rank distance.

A rank code $C$ of length $n$ and dimension $k$ over $GF(q^m)$
is a subspace of dimension $k$ of $GF(q^m)$ viewed as a (rank) metric space.
The minimum rank distance of
the code $C$ is the minimum rank of non-zero vectors of the code. 
 In the following, ${C}$ is a rank metric code of length $n$ and dimension $k$ over
$GF(q^m)$. The matrix $G$ denotes a $k \times n$ generator matrix of $\mathcal{C}$ and $H$ one of its parity check matrices.

\begin{definition} \label{support}
Let $x=(x_1,x_2,\cdots,x_n) \in GF(q^m)^n$ be a vector of rank $r$. We denote $E$ the $GF(q)$-sub vector space of $GF(q^m)$ generated by $x_1,x_2,\cdots,x_n$. The vector space $E$ is called the {\bf support} of $x$.
\end{definition} 

{\bf Remark:}
  The notion of support of a codeword for the Hamming distance and
  for the the one introduced in definition~\ref{support} are different but they share a common principle: in both cases, suppose one is given a syndrome $s$ and that
there exists a low weight vector $x$ such that $H.x^t=s$, then, if the support of $x$ is known,
it is possible to recover all the coordinates values of $x$ by solving a linear system.

\begin{definition} \label{erasure}
Let $e$ be an error vector of rank $r$ and error support space $E$. We call
{\bf generalized erasure of dimension $t$} of the error $e$, a subspace $T$ of dimension $t$
of its error support $E$.
\end{definition}

The notion of erasure for Hamming distance corresponds to knowing a particular
position of the error vector (hence some partial information on the support), 
in the rank distance case, the support of the error
being a subspace $E$, the equivalent notion of erasure (also denoted generalized erasure) is therefore
the knowledge of a subspace $T$ of the error support $E$.

\subsection{Bounds for rank metric codes}

The classical bounds for the Hamming metric have straightforward rank
metric analogues,
since two of them are of interest for the paper we recall them below.

\subsubsection{Rank Gilbert-Varshamov bound [GVR]}
The number of elements $S(m,q,t)$ of a sphere of radius $t$ in $GF(q^m)^n$, is equal to the number
of $m \times n$ $q$-ary matrices of rank $t$. For $t=0$ $S_0=1$, for $t \ge 1$ we have
(see \cite{Loi06}):
$$ S(n,m,q,t)= \prod_{j=0}^{t-1} \frac{(q^n-q^j)(q^m-q^j)}{q^t-q^j}$$

From this we deduce the volume of a ball $B(n,m,q,t)$ of radius $t$ in $GF(q^m)$ to be:
$$B(n,m,q,t)= \sum_{i=0}^tS(n,m,q,i).$$

In the linear case the Rank Gilbert-Varshamov bound $GVR(n,k,m,q)$ 
 for a $[n,k]$ linear code over $GF(q^m)$ is then defined as the smallest integer $t$ such that 
$B(n,m,q,t) \ge q^{m(n-k)}$.

The Gilbert-Varshamov bound for a rank code $C$ with dual matrix $H$
 corresponds to the smallest rank weight $r$ for which, for any 
syndrome $s$, there exists on the average  a word $x$ of rank weight
$r$ such that $H.x^t=s$.
To give an idea of the behaviour of this bound, 
it can be shown that, asymptotically in the case $m=n$ (\cite{Loi06}): 
$\frac{GVR(n,k,m,q)}{n} \sim 1-\sqrt{\frac{k}{n}}$.

\subsubsection{Singleton bound} 

The classical Singleton bound for a linear $[n,k]$ rank code of minimum rank
$r$ over $GF(q^m)$
works in the same way as for Hamming linear codes (by finding an information set)
and reads $r \le 1+n-k$: in the case when $n > m$ this bound
can be rewritten as $r \le 1+ \lfloor \frac{(n-k)m}{n} \rfloor$ \cite{Loi06}.
Codes achieving this bound  are called Maximum Rank Distance codes (MRD).

\subsection{Cryptography and rank codes}\label{sec:rankcodes}

The main use of rank codes in the cryptographic context is through the
rank analogue of the classical syndrome decoding problem.

{\bf Maximum Likelihood - Rank Syndrome Decoding problem (ML-RSD)}
                Let $H$ be an $(n-k)\times n$ matrix over $GF(q^m)$ with $k \le n$, $s \in GF(q^m)^{n-k}$ 
. The problem is to find the smallest weight $r$ such that ${\rm
  rank}(x)=r$ and $Hx^t=s$.

The computational complexity of this problem, which was unkown for more
than 20 years, was recently assessed theoretically in \cite{GZ16}. In
this paper a
randomized reduction to the Hamming distance Syndrome Decoding problem
distance is given. More precisely, it is proved that if there exists a polynomial
algorithm which solves the RSD problem, then NP $\subset$ RP, which is very
unlikely. Results also extend to the case of approximation of
the rank distance of a code by a constant.

Besides the theoretical hardness of the RSD problem, practical attacks on the problem have a complexity which increases very quickly with the parameters.


There exist several types of generic attacks on the problem:

$\bullet$ {\bf combinatorial attacks}: these attacks are usually the best ones for small values
of $q$ (typically $q=2$) and when $n$ and $k$ are not too small
(typically $30$ and more): 
when $q$ increases, the combinatorial aspect makes them less efficient.
The first non-trivial attack on the problem was proposed by Chabaud and Stern \cite{CS96} in 1996,
then in 2002 Ourivski and Johannson \cite{OJ02} improved the previous attack and proposed
a new attack: however, these two attacks did not take account of the value of $n$ in the
exponent. They were generalized recently in \cite{GRS12} by Gaborit et {\it al.}
in $(n-k)^3m^3q^{(r-1)\lfloor\frac{(k+1)m}{n}\rfloor}))$ and 
take the value of $n$ into account and were used to break some
repaired versions of the GPT cryposystem. 

$\bullet$ {\bf algebraic attacks and Levy-Perret attack}: the particular nature of the rank metric makes it a natural field
for algebraic attacks and solving by Groebner basis, since these attack are largely independent of the value of $q$
and in some cases may also be largely independent of~$m$.
These attacks are usually the most efficient when $q$ increases and when 
the parameters are not too high (say less than $30$).
There exist different types of algebraic equations settings: 
the first one by Levy and Perret \cite{LP06} in 2006
considers a quadratic setting by taking as unknowns the support $E$ of the error and the error
coordinates regarding $E$, there is also the Kernel attack by \cite{FLP08} and the minor approach
which consists in considering multivariate equations of degree $r+1$ obtained from minors of matrices
\cite{issac10}, and 
more recently the annulator setting by Gaborit et {\it al.} in \cite{GRS12}
(which is valid on certain type of parameters but may not be independent of $m$).
In our context
for some of the parameters considered in the end of the paper, the Levy-Perret attack is the
most efficient one to consider. The attack works as follows: suppose one starts from
an $[n,k]$ rank code over $GF(q^m)$ and we want to solve the RSD problem for an error 
$e$ of rank weight $r$, the idea of the attack is to consider the support $E$ of $e$
as unknowns together with the error coordinates, it gives $nr+m(r-1)$ unknowns and $m(2(n-k)-1)$
equations from the syndrome equations. One obtains a quadratic system, on which one can use
Groebner basis. All the complexities for Grobner basis attacks
are estimated through the very nice program of L. Bettale \cite{bettale}. 
In practice this attack becomes too costly whenever $r \ge 4$ for not too small
$n$ and $k$.

{\bf The case of more than one solution: approximating beyond the GVR bound}

In code based cryptography there is usually only one solution to the syndrome
problem (for instance for the McEliece scheme), now in this situation
we are interested in the case when there are a large number of solutions.
This case is reminiscent of lattice-based cryptography
when one tries to approximate as much as possible a given syndrome
by a word of weight as low as possible.


This motivates us to introduce a new problem
which corresponds to finding a solution to the general decoding problem
for the case when the weight of the word associated to the syndrome 
is greater than the GVR bound, in that case there may be several solutions,
and hence the term decoding does not seem well chosen.
Notice that in a lattice cryptography context, it corresponds to the
case of Gap-CVP, which does not make sense here, since it implies
a multiplicative gap. 

\medskip

{\bf Approximate - Rank Syndrome Decoding problem (App-RSD)}
                Let $H$ be an $(n-k)\times n$ matrix over $GF(q^m)$ with $k \le n$, $s \in GF(q^m)^{n-k}$
and let $r$ be an integer. The problem is to find a solution of rank  $r$ such that ${\rm rank}(x)=r$ and $Hx^t=s$.

\smallskip

Even though the recent results of \cite{GZ16} show that the problem of approximation of the rank distance remains hard, there are cases for which the problem is easy, that we want to consider.

It is helpful to first consider the situation of a binary linear $[n,k]$
Hamming metric code.
Given a random element of length $n-k$ of the syndrome space,
we know that with high probability there exists a word that has this
particular syndrome and whose weight is on
the GV bound.
This word is usually hard to find, however. Now what is the lowest minimum weight
for which it is easy to find such a word ? A simple approach consists
in taking $n-k$ random column of the parity-check matrix (a potential
support of the solution word) and inverting 
the associated matrix, multiplying by the syndrome gives us
a solution of weight $(n-k)/2$ on average. In fact 
it is difficult to do better than this without a super-polynomial increase
in complexity.

Now for the rank metric, one can apply the same approach: suppose
one starts from a random $[n,k]$ code over $GF(q^m)$ and that one searches 
for a word of small rank weight $r$ with a given syndrome. One fixes
(as in the Hamming case)
a potential support for the word - here a subspace of dimension $r$ of $GF(q^m)$-
and one tries to find a solution. Let $x=(x_1,\cdots,x_n)$ be a solution vector,
so that $H.x^t=s$. If we consider the syndrome equations induced in the small
field $GF(q)$, there are $nr$ unknowns and $m(n-k)$ equations.
Hence it is possible (with a good probability)
to solve the system whenever $nr \ge m(n-k)$, meaning
it is possible to find in probabilistic polynomial time a solution to
a typical instance of the RSD problem
whenever $r \ge  \lceil \frac{m(n-k)}{n} \rceil$, which corresponds 
to the Singleton bound. This proves the following proposition:

\begin{proposition}
There is a probabilistic polynomial time algorithm that solves random
instances of the App-RSD problem in polynomial time when 
 $r \ge  \lceil \frac{m(n-k)}{n} \rceil$.
\end{proposition}

For a rank weight $r$ below this bound, the best known attacks are, as
in the Hamming distance case,
obtained by considering the cost of finding a word of rank $r$ divided by the 
number of potential solutions: $\frac{B(n,k,m,q)}{q^{m(n-k)}}$.
In practice the complexity we find is coherent with this.

\section{Approximating a random syndrome beyond the GVR bound with LRPC codes}
\label{sec:approximating}
\subsection{Decoding algorithm in rank metric}
The rank metric has received a lot of attention in the context of network coding \cite{SKK10}.
There exist very few algorithms, however, for decoding codes in the rank metric. The most well-known
$[n,k]$ codes which are decodable are the Gabidulin codes
\cite{Gab85}. These codes can correct 
up to $\frac{n-k}{2}$ errors, and have been proposed for encryption: but since they cannot decode up to the 
GVR bound, they do not seem suitable for full decoding in the spirit of \cite{CFS} for signature algorithms. Another more recent 
family of decodable codes are the LRPC codes \cite{GMRZ13}, these
codes are defined through a low rank matrix.

\begin{definition}
  A Low Rank Parity Check (LRPC) code of rank $d$, length $n$ and
  dimension~$k$ over $GF(q^m)$ is a code defined 
 by an $(n-k) \times n$ parity check matrix $H=(h_{ij})$, such that all its
coordinates $h_{ij}$ belong to the same $GF(q)$-subspace $F$ of dimension $d$ of $GF(q^m)$. 
We denote by $\{F_1,F_2,\cdots,F_d\}$ a basis of $F$.
\end{definition}

These codes can decode with a good probability up to $\frac{n-k}{d}$ errors, they can be used for encryption
\cite{GMRZ13}, but since they can decode only up to  $\frac{n-k}{2}$
errors at best, they also seems unsuitable for signature algorithms.

\subsection{Using LRPC codes to approximate a random syndrome beyond the GVR bound}\label{sec:UsingLRPC}

\subsubsection{High level overview}

The traditional approach for decoding random syndromes, that is used by the CFS scheme for instance,
consists in taking advantage of the decoding properties of a code
(e.g. a Goppa code) and in considering parameters for which the
proportion of decodable vectors -- the decodable density -- 
is not too low. For the Hamming metric, this approach leads to 
very flat dual matrices, i.e., codes with high rate and very low Hamming
distance.
In the rank metric case, this approach leads to very small decodable
densities and does not work in practice.
However, it is possible to proceed otherwise. 
It turns out that the decoding algorithm of LRPC codes can be adapted so that 
it is possible to decode not only errors but also (generalized) erasures.
This new decoding algorithm allows us to decode more rank errors since
the support is then partially known.
In that case since the size of the balls depends directly on the dimension
of the support, it leads to a dramatic increase of the size of the decodable balls.
Semantically, what happens is that the signer can fix an erasure space, which relaxes
the condition for finding a preimage.
 This approach works because in the particular case of 
our algorithm, it is possible to consider the erasure space at no cost 
in terms of error correction: to put it differently, 
the situation for LRPC is different from traditional Hamming metric
codes for which ``an error equals two erasures''.

In practice it is possible to find parameters (not flat at all) for which it is possible
to decode a random syndrome with the constraint that its support contains a fixed random
subspace. Fixing part of the rank-support of the error, (the generalized
erasure) allows us more rank-errors.
For suitable parameters, the approach works then as follows: for a given 
random syndrome-space element $s$, one chooses a random subspace $T$ of fixed
dimension $t$ (a generalized erasure of Definition 2), 
and the algorithm returns a small rank-weight word, whose rank-support $E$ contains $T$,
and whose syndrome is the given element $s$. Of course, there is no unicity 
of the error $e$ since different choices of $T$ lead to different errors $e$,
which implies that the rank of the returned error is above the GVR
bound: it is however only just above the GVR bound for the right choice of parameters.

\subsubsection{LRPC decoding with errors and generalized erasures}

\paragraph{Setting:}
Let  an $[n,k]$ LRPC code be defined by an $(n-k)\times n$ parity-check matrix
$H$ whose entries lie in a space $F\subset GF(q^m)$ of small dimension
$d$. Let $t$ and $r'$ be two parameters such that
$$r' \le \frac{n-k}{d}.$$
Set $r=t+r'$. 
Given an element of the syndrome
space $s$, we will be looking for a rank $r$ vector $e$ of $GF(q^m)^n$
with syndrome $s$. We first look for an acceptable subspace $E$ of 
dimension $r$ of $GF(q^m)$ and then solve the linear system $H.e^t=s$
where $e\in E^n$. To this end we choose a random subspace $T$ of
dimension $t$ of $GF(q^m)$ and impose the condition $T\subset E$.

The subspace $T$ being fixed,
we now describe the set of decodable elements of the syndrome space.
We will then see how to decode them.

\begin{definition}\label{def:Tdecodable}
  Let $F_1$ and $F_2$ be two fixed linearly independent elements of the space $F$.
  We shall say that an element $s\in GF(q^m)^{n-k}$ of the syndrome
  space is {\em $T$-decodable} if there exists a rank $r$
  subspace $E$ of $GF(q^m)$
  satisfying the following conditions.
  \begin{enumerate}
\item[{\rm (i)}] $\dim \langle FE\rangle=\dim F\dim E$,
\item[{\rm  (ii)}] $\dim (F_1^{-1}\langle FE\rangle\cap F_2^{-1}\langle FE\rangle) =\dim E$,
\item[{\rm  (iii)}] the coordinates of $s$ all belong to the space $\langle
  FE\rangle$ and together with the elements of the
space $\langle FT\rangle$ they generate the
  whole of $\langle FE\rangle$.
\end{enumerate}
\end{definition}

\paragraph{Decoding algorithm.}
We now argue that if a syndrome $s$ is $T$-decodable, we can
effectively find $e$ of rank $r$ such that $H.e^t=s$. We first
determine the required support space~$E$. Since the decoder knows the
subspaces $F$ and $T$, he has access to the product space $\langle
FT\rangle$. He can then construct the subspace $S$ generated by
$\langle FT\rangle$ and the coordinates of $s$. Condition (iii) of
$T$-decodability ensures that the subspace $S$ is equal to $\langle
FE\rangle$
for some $E$, and since 
$$F_1^{-1}\langle FE\rangle\cap F_2^{-1}\langle FE\rangle \supset E,$$
condition (ii) implies that $E$ is uniquely determined and that the
decoder recovers $E$ by computing the intersection of subspaces
$F_1^{-1}S\cap F_2^{-1}S$.

It remains to justify that once the subspace $E$ is found, we can
always find $e$ of support $E$ such that $H.e^t=s$. This will be the
case if the mapping
\begin{eqnarray}
    E^n& \rightarrow & \langle FE\rangle^{n-k}\label{eq:mappingEn}\\
      e& \mapsto     & H.e^t \nonumber
\end{eqnarray}
can be shown to be surjective.
Extend $\{F_1,F_2\}$ to a basis $\{F_1,\cdots,F_d\}$ of $F$
and let  $\{E_1,\cdots,E_r\}$ be a basis of $E$.
Notice that the system $H.e^t=s$ can be rewritten formally
as a linear system in the small field $GF(q)$ where the coordinates
of $e$ and the elements of $H$ are written 
in the basis $\{E_1,\cdots,E_r\}$ and $\{F_1,\cdots,F_d\}$
respectively, and where the syndrome coordinates are written in the product basis
 $\{E_1.F_1,\cdots,E_r.F_d\}$. We therefore have a
linear system with $nr$ unknowns
and $(n-k)rd$ equations over $GF(q)$ that is defined by an
$nr \times (n-k)rd$ formal matrix $H_f$ ({\it say})
whose coordinates are functions only of $H$ (see \cite{GMRZ13} for more details on how to obtain $H_f$ from $H$). 

We now see that the
matrix $H$ can be easily chosen so that the matrix $H_f$ is of maximal
rank $nr$, which makes the mapping \eqref{eq:mappingEn} surjective,
{\em for any subspace $E$} of dimension $d$ satisfying condition (i)
of $T$-decodability.

{\bf Remarks:}
\begin{enumerate}
\item For applications, we will consider only the case where
  $nr=(n-k)rd$, meaning that the mapping \eqref{eq:mappingEn} is
  always one-to-one.
\item The system $H.e^t=s$ can be formally inverted and stored in a
  pre-processing phase, so that the decoding complexity is only that
  of multiplication by a square matrix of length $nr$, rather than a
  cubic inversion.
\item In principle, the decoder could derive the support $E$ by
  computing
  \begin{equation}
    \label{eq:E1toEd}
    E=F_1^{-1}S \cap \cdots \cap F_d^{-1}S
  \end{equation}
  rather than simply $E=F_1^{-1}S \cap F_2^{-1}S$, and the procedure
  would work in the same way in cases when \eqref{eq:E1toEd} holds but
  not the simpler condition (ii). This potentially increases the set
  of decodable syndromes, but the gain is somewhat marginal and
  condition (ii) makes the forthcoming
  analysis simpler. For similar reasons, when conditions (i)--(iii)
  are not all satisfied, we do not attempt to decode even if there are
  cases when it stays feasible.
\end{enumerate}

Figure \ref{fig:Decoding algorithm of LRPC codes with a fixed partial
  support $T$}
summarizes the decoding algorithm. Note that the decoder can easily
check conditions (i)--(iii), and that a decoding failure is declared
when they are not satisfied.

\begin{figure}[!h]
\centering
\fbox{
  \begin{minipage}{12cm}
{\bf Input}: $T=\langle T_1,\cdots, T_t \rangle$ a subspace of
$GF(q^m)$ of
dimension $t$, $H$ an $(n-k)\times n$ matrix with elements in a
subspace 
$F =\langle F_1,\cdots, F_d \rangle$
of dimension $d$, and $s\in GF(q^m)^{n-k}$.\\
{\bf Output}: a vector $e=(e_1,\ldots e_n)$ such that $s=H.e^t$, 
with $e_i\in E$, $E$ a subspace of dimension $\dim E = r =
t+\frac{n-k}{d}$ satisfying
$T\subset E$.
\begin{enumerate}
\item {\bf Syndrome computations}

a) Compute a basis $B=\{F_1T_1,\cdots, F_dT_t\}$ of the product space
$\langle F.T\rangle$.

b) Compute the subspace $S=\langle B \cup \{s_1,\cdots,s_{n-k}\}\rangle $.

\medskip
\item {\bf Recovering the support $E$ of the error}

Compute the support of the error $E = F_1^{-1}S \cap F_2^{-1}S$,
and compute a basis $\{E_1,E_2,\cdots,E_{r}\}$ of $E$.
\medskip
\item {\bf Recovering the error vector $e=(e_1,\ldots ,e_n)$}

For $1 \le i \le n$, write $e_i=\sum_{i=1}^ne_{ij}E_j$,
solve the system $H.e^t=s$, where the equations $H.e^t$ and the syndrome coordinates $s_i$
are written as elements of the product space $P=\langle E.F\rangle$ in the basis
$\{F_1E_1,\cdots,F_1E_r,\cdots,F_dE_1,\cdots,F_dE_r\}$. The system has $nr$ unknowns (the $e_{ij}$) in $GF(q)$ and $(n-k)rd$ equations from the syndrome.
 \end{enumerate}
  \end{minipage}
} \caption{Algorithm 1: a general errors/erasures decoding algorithm
  for LRPC codes}
\label{fig:Decoding algorithm of LRPC codes with a fixed partial support $T$}
\end{figure}
 
\subsection{Proportion of decodable syndromes for unique decoding of
  LRPC codes}

Signature algorithms based on codes all inject the message space in
some way into the syndrome space and then decode them to form a
signature. We should therefore estimate the proportion of decodable
syndromes. The classical decoding approach tells us to look for a preimage
by $H$ that sits on the Gilbert-Varshamov bound: for typical random
codes, a preimage typically exists and is (almost) unique. Computing
such a preimage is a challenge, however. In our case, we are looking
for a preimage above the Gilbert-Varshamov bound, for which many
preimages exist, but for a fixed (erasure) subspace $T$, decoding
becomes unique again. In the following, we count the number of
$T$-decodable syndromes and show that for some adequate parameter
choices, their proportion can be made to be close to $1$.
It will be convenient to use the following notation.

\begin{definition}
  For a subspace $T$ of $GF(q^m)$ of dimension $t$, denote by
  $\mathcal{E}(T)$ the number of subspaces of dimension $r=r'+t$ that
  contain $T$.
\end{definition}

\begin{lemma}
We have
$$\mathcal{E}(T) = \prod_{i=0}^{r'-1}\left(\frac{q^{m-t-i}-1}{q^{i+1}-1}\right)$$
\end{lemma}
\begin{proof}
Consider the case where $r=t+1$, we need to construct distinct subspaces of dimension
$t+1$ containing $T$. This can be done by adjoining an element of $GF(q^m)$ modulo the subspace
$T$, which gives $(q^m-q^t)/(q^{t+1}-q^t)=(q^{m-t}-1)/(q-1)$
possibilities. Now any subspace of dimension $t+1$
contains $q^{t+1}-1$ supspaces of dimension $t$ containing $T$. A repetition
of this approach $r'-1$ times gives the formula. (see also \cite{McWSlo77} p.630).
\end{proof}

\begin{theorem}\label{thm:density}
The number ${\mathcal T}(t,r,d,m)$ of $T$-decodable syndromes
satisfies the upper bound:
$${\mathcal T}(t,r,d,m)\leq \mathcal{E}(T)q^{rd(n-k)}.$$
Furthermore, under the conditions
$r(2d-1) \le m$ and 
\begin{align}
  \label{eq:FT}
  \dim\langle FT\rangle &= \dim F\dim T,\\
  \label{eqa:2dimF-1}
  \dim(F_1^{-1}F + F_2^{-1}F) &=
2\dim F -1 = 2d-1,
\end{align}
we also have the lower bound:
 $$\left(1-\frac{1}{q-1}\right)^2\mathcal{E}(T)q^{rd(n-k)}\leq {\mathcal T}(t,r,d,m).$$
\end{theorem}

Note that condition \eqref{eqa:2dimF-1} depends only on the subspace
$F$ and can be ensured quite easily when designing the matrix
$H$. Random spaces $F$ with random elements $F_1$ and $F_2$ will
typically have this property. Condition \eqref{eq:FT} depends on the
choice of the subspace $T$: as will be apparent from Lemma below, 
for a random subspace $T$ condition \eqref{eq:FT} holds with probability very
close to $1$.

The complete proof of Theorem \ref{thm:density} is given in Appendix
\ref{app:density}.

\pagebreak

{\bf Remarks:}
\begin{enumerate}
\item
It can be shown with a finer analysis that the term $(1-1(q-1))^2$ in the lower
bound can be improved to a quantity close to $1-1(q-1)$.
\item For large $q$, Theorem \ref{thm:density} shows that, for most
  choices of $T$, the density
  of $T$-decodable syndromes essentially equals 
  \begin{equation}
    \label{eq:density}
    \frac{\E(T)q^{rd(n-k)}}{q^{m(n-k)}} \approx
    q^{(r-t)(m-r)+(n-k)(rd-m)}.
  \end{equation}
\end{enumerate}

Remarkably, it is possible to choose sets of parameters $(m,t,r,d)$,
with $(n-k)=d(r-t)$, such that the exponent in \eqref{eq:density}
equals zero, which gives a density very close to $1$.

{\bf Example of parameters with density almost $1$:}  For $q=2^{8},m=18,n=16,k=8,t=2,r'=4$, the algorithm
decodes up to $r=t+r'=6$ for a fixed random partial support $T$ of dimension $2$. 
The GVR bound for a random $[16,8]$ code with $m=18$ is $5$, the Singleton bound is $8$,
we see that the decoding radius $6$ is therefore just above the GVR bound at $5$ and smaller 
than the Singleton bound at $8$. Moreover one can notice that
if parameters $(m,t,r,d)$ satisfy the two equations $(r-t)(m-r)+(n-k)(rd-m)=0$
and $(n-k)=d(r-t)$ (the case for which the density is almost $1$), then
for any integer $\alpha$ greater than $1$, the parameter set $(\alpha m,\alpha t,\alpha r,d)$
satisfies the same  equations, and hence for a given $d$ one obtains 
an infinite family of parameters with density almost $1$.

{\bf Decoding in practice.}
In practice it easy enough to find sets of parameters for which the density of decodable
syndromes is very close to $1$, i.e. such that
$(r-t)(m-r)+(n-k)(rd-m)=0$. 

\section{RankSign, a signature scheme for the rank metric based on augmented LRPC codes}
\label{sec:RankSign}
We saw in the previous section how to construct a matrix $H$ of an
LRPC code, with
a unique support decoding, which opens the way for a signature algorithm.
In practice the best decoding results are obtained for $d=2$: the
natural strategy is to define for
the public key a matrix $H'=AHP$, where
$A$ is a random $(n-k) \times (n-k)$
invertible matrix in the extension
field and $P$ is an invertible $n \times n$ matrix in the small field.
However,
it is easily possible for a cryptanalyst to recover the words of small
weight $d=2$ in $H'$ and
it is therefore necessary to hide the matrix $H$ in another way. In
what follows we present 
a simple type of masking: $RankSign$ which consists in adding a few random columns to $H$.

Suppose one has a fixed support $T$ of dimension $t$.
We consider the public matrix $H'=A(R|H)P$ with $R$ a random $(n-k) \times t'$ matrix in $GF(q^m)$.
We will typically take $t'=t$ but one could envisage other values of $t'$. 
We call {\it augmented LRPC codes} such codes with
parity-check matrices $H'=A(R|H)P$.

Starting from a partial support $T$ that has been randomly chosen and
is then fixed,
the signature consists in decoding not a random $s$ but the syndrome
$s'=s-R.(e_1,\cdots,e_t)^t$ for $e_i$ random independent elements of $T$. 

The overall rank of the solution vector $e$ is still $r=t+r'$.
the masking gives us that the minimum rank-weight of the code
generated by the rows of $H'$ is $t+d$ 
rather than purely $d$: therefore recovering the hidden structure
involves finding relatively large minimum weight vectors in a code.
In practice we consider
$d=2$ and $H$ is a $n/2 \times n$ matrix with all coordinates in a
space $F$ of dimension $2$.
Moreover for $\{F_1,F_2\}$  a basis of $F$, we choose the matrix $H$ such
that when $H$ is written in the basis  $\{F_1,F_2\}$, one obtains a $n \times n$ 
invertible matrix (of maximal rank) over $GF(q)$. It can be done
easily. Figure~\ref{fig:Signature} describes
the scheme, where || denotes concatenation. 

\begin{figure}[!h]
\centering
\boxed{
  \begin{minipage}{12.5cm}
\begin{enumerate}
\item {\bf Secret key:} an augmented LRPC code over $GF(q^m)$ 
with parity-check matrix $(R|H)$ of size $(n-k) \times (n+t)$
                     which can decode $r'$ errors and $t$ generalized
                     erasures: a randomly chosen
$(n-k) \times (n-k)$ matrix $A$ that is invertible in $GF(q^m)$  a
randomly chosen $(n+t) \times (n+t)$
matrix $P$ invertible in $GF(q)$.  

\item {\bf Public key:} the matrix $H'=A(R|H)P$, a small integer value $l$, a hash function $hash$.

\item {\bf Signature of a message $M$:}  

     a) {\it initialization}: seed $\leftarrow$ $\{0,1\}^l$, pick $t$ random independent elements $(e_1,\cdots,e_t)$ of $GF(q^m)$

     b) {\it syndrome}: $s \leftarrow hash(M||seed) \in GF(q^m)^{n-k}$      
   
     c) decode by the LRPC matrix $H$, 
the syndrome $s'=A^{-1}.s^T-R.(e_1,\cdots,e_t)^T$ with erasure space 
$T=\langle e_1,\cdots,e_t\rangle$ and $r'$ errors by Algorithm 1.

d) if the decoding algorithm works and returns a word $(e_{t+1},\cdots,e_{n+t})$ of weight $r=t+r'$,
signature=$((e_1,\cdots,e_{n+t}).(P^T)^{-1}, seed)$, else return to a).

\item {\bf Verification}: Verify that $Rank(e)=r=t+r'$ and $H'.e^T=s=hash(M||seed)$.
 
\end{enumerate}
  \end{minipage}
} \caption{The $RankSign$ signature algorithm}
\label{fig:Signature}
\end{figure}

{\bf Parameters:}  {\it Public key size:} $(k+t)(n-k)mLog_2(q)$  {\it Signature size:} $(m+n+t)rLog_2(q)$.

The cost of the decoding algorithm is quadratic because of preprocessing
of $H_f^{-1}$, hence the major cost comes from the linear algebra over the large
field $GF(q^m)$. 

{\it Signature complexity:} $(n-k) \times (n+t)$ operations in $GF(q^m)$. {\it Verification complexity:} $(n-k) \times (n+t)$ operations in $GF(q^m)$.

The length $l$ of the seed can be taken equal to $\frac{80}{Log_2(q)}$ for instance.

\section{Security analysis of the scheme }

\subsection{Security of augmented LRPC codes}\label{sec:security}

In the previous section we defined augmented-LRPC with dual matrix
 $H'=A(R|H)P$, we now formulate the problem Ind-LRPC codes (Ind-LRPC) on the security of these
codes:

\medskip
{\bf Problem [Ind-LRPC] {\it The augmented LRPC codes are indistinguishable from random codes.}
}
\smallskip

We know make the following assumption that we will
discuss below:

\smallskip

{\bf Assumption:}{ \it the Ind-LRPC problem is difficult.}

\smallskip

{\bf {\it Discussion on the assumption:}} The family of augmented LRPC codes is not of course 
a family of random codes, but they are weakly structured codes: the main point being
that they have a parity-check matrix one part of which consists only
of low rank coordinates the other part consisting of
random entries.
The attacker never has direct access to the LRPC matrix $H$, which is hidden
 by the augmented part.

The minimum weight of augmented LRPC codes is smaller than the GVR bound,
hence natural attacks consist in trying to use their special structure
to attack them. There exist general attacks for recovering the minimum weight of a code
(see Section \ref{sec:rankcodes}) but these attacks have a fast increasing complexity especially
when the size of the base field $GF(q)$ increases.
We first list obvious classical attacks for recovering the structure of augmented-LRPC codes
and then describe specific attacks.

$\bullet$ {\it Previously known structural attacks for rank codes.} The main structural attack 
for the rank metric is the Overbeck attack on the GPT cryptosystem.
The attack consists in considering concatenated public matrices $G^q,G^{q^2},...,G^{q^{n-k-1}}$,
in that case the particular structure of Gabidulin codes enables one to find a concatenated
matrix with a rank default; this is due to the particular structure of the Gabidulin codes
and the fact that for Gabidulin codes $G^{q^i}$ is very close to $G^{q^{i+1}}$.
In the case of LRPC codes, since the rows are taken randomly in a small space, this
attack makes no sense, and cannot be generalized. 

\medskip

$\bullet$ {\it Dual attack: attack on the dual matrix $H'$.} Another
approach consists in directly finding words of small weight induced
by the structure of the code, from which one can hope to recover the global 
structure. 
For augmented LRPC codes, the rank of the small weight
words is $d+t$: $d$ for LRPC and $t$ for the masking. This attack becomes very hard
when $t$ increases, even for low $t$. For instance for $t=2$ and $d=2$ it gives a minimum weight
of $4$, which for most parameters $n$ and $k$ is already out of reach
of the best known attacks
on the rank syndrome decoding problem (see Section \ref{sec:rankcodes}).   

\medskip

$\bullet$ {\it Attack on the isometry matrix P.} Remember that for rank metric codes,
the isometry matrix {\it is not a permutation matrix} but an invertible matrix over the 
base field $GF(q)$. The attacker can then try to guess
the action of $P$ on $H$, since $d$ is usually small
negating this action may allow to attack directly a code of rank $d$.
Since $d$ is small it is enough to guess the resulting action of $P$ on $n-k+3$ 
columns by considering only the action of $P$ coming from the first $t$ columns 
of the matrix R -- the only columns which may increase the rank -- it means
guessing $(n-k+3)\times t$ elements of $GF(q)$ (since coordinates of $P$ are in $GF(q)$), 
hence a complexity of $q^{(n-k+3)t}$.
In general this attack is not efficient as soon as $q$ is not small (for instance $q=256$).

\medskip

$\bullet$ {\it Attack on recovering the support.} An attacker may also 
try to recover directly an element of the support. For instance
in the case of $d=2$, for $F$ the error support generated by $\{F_1,F_2\}$,
up to a constant one can rewrite $F$ as generated by $1$ and $F_2.F_1^{-1}$.
Then the attacker can try to guess the particular element $F_2.F_1^{-1}$,
recover $F$ and solve a linear system in the coordinates of the elements 
of $H$. The complexity of this attack is therefore $q^m.(nd)^3$. Even in the most favourable
case when $d=2$ this attack
is exponential and becomes infeasible for $q$ not too small.

\medskip

$\bullet$ {\it Differential support attack.}
It is also possible to search for an attack directly based on the specific structure
of the augmented LRPC codes.
The general idea of the differential support attack  is to consider the vector space 
$V$ on the base field $GF(q)$ generated by the elements of a row of the augmented matrix $H'$
and to find a couple $(x,x')$ of elements of $V$ such that $\frac{x'}{x} \in F$ the support
of the LRPC code. The complexity of the attack is at least $q^{(n-k)(d-1)+t}$,
the detail of the attack can be found in Appendix~\ref{app:differential}. In practice this exponential attack is often
the best attack for recovering the structure of the code and distinguishing the augmented
LRPC code from a random code.

\medskip

{\bf Conclusion on the hardness of the Ind-LRPC problem}

Even though there are many possible strategies for attacking the Ind-LRPC problem, in particular 
because of the rich structure of rank metric, the above discussion of general known
attacks shows that they are all exponential with a strong dependency
on the size of $q$. Moreover, we also considered very specific attacks (like
the differential support attack) related
to the particular structure of the augmented LRPC codes. 
This analysis seems to show that the Ind-LRPC 
problem is indeed difficult, with all known attacks being exponential. 
In practice  it is easy to find parameters which resist all these attacks.
 

\subsection{Information leakage}
The attacks considered above concerned
the case where no additional information
was known beside the public parameters.
Often the most efficient attacks on signatures is to recover the hidden structure
of the public key by using information leaking from real signatures. This for instance
is what happened in the case of NTRUSign: the secret key is not directly attacked,
but the information leaked from real signatures enables one to recover successfully
the hidden structure. We show below that
with our masking scheme no such phenomenon can occur,
since we prove that, if an attacker can break the signature scheme
for public augmented matrices with the help of information leaking
from a number of (approximately) $q$ real signatures, then he can also
break the scheme just as efficiently {\it *without*} any authentic
signatures.

Theorem~\ref{thm:forgery} below states the unleakibility of signatures.
It essentially states that valid signatures leak no information on
the secret key. More precisely, there
exists a polynomial time probabilistic algorithm that takes as input the
public matrix $H'$ and  produces couples
$(m,\sigma)$,
where $m$ is a message and $\sigma$ a valid signature
for $m$ and that, under the random oracle model, has the
same probability distribution as couples (message, signature)
output by the
authentic signature algorithm, and is therefore indistinguishable from
them. Therefore, whatever forgery can be
achieved from the knowledge of $H'$ and a list of valid signed
messages, can be simulated and reproduced with the public matrix
$H'$ as only input.

\begin{theorem}:\label{thm:forgery}
 For any algorithm $\A$ that leads to a forged signature using $N\leq q/2$
 authentic signatures, there is an algorithm $\A'$ with the same complexity
 that leads to a forgery using only the public key as input and
 without any authentic signatures.
\end{theorem}

\proof see Appendix~\ref{app:leakage}.
\qed

\subsection{Unforgeability}

Our main Theorem~\ref{thm:forgery} and its proof, show that it is possible to simulate (message,signature) couples with the same probability distribution as valid (message,signature) couples whenever the number of such couples is less than $q/2$. Therefore, given less than $q/2$ signatures (chosen or given), an attacker cannot do better than an attacker who knows only the public key (the matrix of a code). And in that case, under the Ind-LRPC indistinguashability assumption of augmented LRPC codes with random codes, it implies that forging a false signature in the ROM (i.e. being able to approximate a random syndrome for the augmented LRPC class of
 codes) means being able to decode a random rank code. Parameters of the scheme are hence chosen with large $q$ and suitable code parameters for which it is difficult to decode a random code and to distinguish augmented LRPC codes from random codes.


\section{Practical security and parameters}
Below we give in Table~\ref{tab:parameters} some examples of parameters. The parameters
are adjusted to resist all previously known attacks. 
The security reduction  holds for up to $q/2$ 
signatures, hence if one considers $q=2^{40}$ it means we are
protected against leakage
for up to $2^{40}$ obtained authentic signatures. Such an amount of signatures
is very difficult to obtain in real life, moreover if one multiplies
by the amount of time necessary to obtain a signature (about $2^{30}$
for $q=2^{40}$) we clearly see that
obtaining such a number of authentic signatures is out of reach, and it 
justifies our security reduction.

We also give parameters for $q$ lower than $2^{40}$: in that case the reduction
is weaker in the sense that it does not exclude a leaking attack
for sufficiently many signatures. However, such a leaking attack 
seems difficult to obtain anyway, and these parameters can be seen as challenges
for our system. 

In the table the considered codes are $[n+t,k+t]$ codes which give a signature
of rank~$r$. The dual code $H'$ is a $[n+t,n-k]$ code which contains words of rank
$d+t$. In the table `LP' stands for the logarithmic complexity of the algebraic Levy-Perret
attack, for instance in the case $n=16$, one gets a $[18,8]$ code in which one searches
for words of rank $4$, it gives $270$ quadratic equations for 126 unknowns, with a theoretical complexity
of $2^{120}$ from~\cite{bettale} (remember that for a random quadratic system over $GF(2)$
with $n$ unknowns and $2n$ equations the complexity is roughly $2^n$ operations
in the base field $GF(2)$). The complexity of a direct attack for searching
low weight words of weight $d+t$ with combinatorial attacks (see Section~\ref{sec:rankcodes})
is given in `Dual'. Finally, `DS' stands for the differential support attack 
of Section~\ref{sec:security} and `DA' stands for the direct attack on the signature in which one searches
directly for a forgery 
for a word of weight $r$ in a $[n+t,k+t]$ code. In the table the number of augmented
columns is usually $t$ except for the last example for which one adds $2$ columns
rather than $t=5$.

The analysis of the security complexities shows that the best attack (in bold in the table)
depends on the given parameters: when $q$ is large the algebraic attacks are better
since they do not really depends on $q$, when $d$ increases the decoding algorithm
is less efficient and then one get closer to the Singleton bound and direct
forgery for the signature becomes easier. For other parameters, usually the specific 
structural differential support attack DS is better. 

\begin{table}
\begin{center}
{\scriptsize
\begin{tabular}{|c|c|c|c|c|c|c|c|c|c|c|c|c|c|c|c|}
\hline
n & \hspace{-2mm} n-k \hspace{-2mm} & m & q & d& t & r' & r & \hspace{-2mm} GVR \hspace{-2mm} & \hspace{-2mm} Singleton \hspace{-2mm} & \hspace{-2mm} pk(bits) \hspace{-2mm} & \hspace{-2mm} sign(bits) \hspace{-2mm} & LP & Dual & DS & DA  \\
\hline
16 & 8  & 18 & $2^{40}$ & 2 &2  &  4 & 6 & 5 & 8 & 57600 & 8640 & {\bf 130} &1096  & 400  &776 \\
\hline
16 & 8  & 18 & $2^8$ & 2  &2& 4 & 6  &5 & 8 &11520 & 1728 & 110 & 233 & {\bf 80} &168 \\
\hline
16 & 8  & 18 & $2^{16}$ & 2&2  & 4 & 6 & 5  & 8 &23040 & 3456 & {\bf 120} & 448 & 160 & 320 \\
\hline
20 & 10 & 24 & $2^8$ & 2 &  3 &  5 & 8 & 6 & 10 & 24960 & 3008 & 190 &370 & {\bf 104} &226 \\
\hline 
27 & 9  & 20 & $2^{6}$ & 3 & 2 & 3 & 5 & 4 & 7 & 23328 & 1470 & 170 &187 &{\bf 120} &129 \\
\hline
48 & 12  & 40 & $2^4$ & 4 & 5 & 3 &  8 & 6 & 10 & 78720 & 2976 & >600  &340 & 164 & {\bf 114} \\
\hline
50 & 10  & 42 & $2^4$ & 5  & 5(2)& 2&  7 & 5 & 9 &70560 & 2800 &  >600  & 240 & 180  & {\bf 104}\\
\hline

\end{tabular}
}
\end{center}
\caption{Examples of parameters for the RankSign signature scheme}
\label{tab:parameters}
\end{table}

{\bf Implementation:}
We implemented our scheme in a non optimized way, the results we obtained showed
that for small $q$ the scheme was very fast, when $q$ increases, one has to consider
the cost of multiplication in $GF(q)$, however for $q=2^{8}$ or $q=2^{16}$ some
optimized implementation may reduce this cost.

\section{Conclusion}
We have introduced a new approach to devising signatures with coding theory
and in particular in the rank metric, by proposing to decode both erasures and errors
rather than errors only.  This approach enables one to return a small weight word
beyond the Gilbert-Varshamov bound rather than below. We proposed a new efficient algorithm
for decoding LRPC codes which makes this approach feasible.
We then proposed a signature scheme based on this algorithm and the full decoding
of a random syndrome beyond the Gilbert-Varshamov bound. We also showed
that it was possible to protect our system against leakage from authentic signatures.
Finally, we propose different types of parameters, some of which are decently small.
The parameters we propose compare very well to other existing signature 
schemes based on coding theory such as the CFS scheme for instance.

\bibliographystyle{plain}

\appendix

\section{Proof of Theorem~\ref{thm:density}}\label{app:density}

We give here a complete proof of Theorem~\ref{thm:density}.
To prove the theorem  we rely on the following
lemma: 

\begin{lemma}\label{lem:AB}
Let $A$ be a fixed subspace of $GF(q)^m$ of dimension $\alpha$ and 
  let $T$ be a subspace of dimension $t$ (with possibly $t=0$) such
  that $\dim\langle AT\rangle = \alpha t$.
  Let $B$ be a subspace
  generated by $T$ together with $\beta$ random independent uniform
  vectors, with $\beta$ satisfying $\alpha(t+\beta)\leq m$.
 Then
  $$\prob{\dim\langle AB\rangle < \alpha(t+\beta)}\leq\frac{q^{\alpha(t+\beta)}}{(q-1)q^m}.$$
\end{lemma}

\begin{proof}
  Suppose first that $B = B'+\langle b\rangle$ where $b$ is a uniformly chosen
  random element of $GF(q)^m$ and where $B'\supset T$ is a fixed space such that
  $\dim\langle AB'\rangle = \alpha(t+\beta -1).$
  Let $A^P$ be a projective version of $A$, meaning that for every
  $a\neq 0$ in $A$, we have exactly one element of the set
  $$\{\lambda a, \lambda\in GF(q)^*\}$$
  in $A^P$.

  We have $\dim\langle AB\rangle < \alpha(t+\beta -1) +\alpha$ if and only
  if the subspace $bA$ has a non-zero intersection with $\langle
  AB'\rangle$, and also if and only if the set $bA^P$ has a non-zero
  intersection with $\langle AB'\rangle$.
Now,
\begin{align}
  \prob{\dim\langle AB'\rangle \cap Ab \neq\{0\}} &\leq \sum_{a\in A^P\!,\,
    a\neq 0} \prob{ab\in \langle AB'\rangle}\nonumber\\
  &= \frac{|A|-1}{q-1}\frac{q^{\alpha(t+\beta-1)}}{q^m}\nonumber\\
  &= \frac{q^{\alpha(t+\beta)}}{(q-1)q^m}-\frac{q^{\alpha(t+\beta-1)}}{(q-1)q^m}.\label{eq:AB'}
\end{align}
since for any fixed $a\neq 0$, we have that $ab$ is uniformly
distributed in $GF(q)^m$, and since the number of elements in $bA^P$
equals $(|A|-1)/(q-1)$.

Now write $$B_0=T\subset B_1=T+\langle b_1\rangle \subset B_2=T+\langle
b_1,b_2\rangle\subset\cdots ,\subset B_i=T+\langle b_1,\ldots
,b_i\rangle\subset \cdots \subset B=B_\beta$$
where $b_1\ldots ,b_\beta$ are independent uniform vectors in $GF(q)^m$.
We have that the probability 
$$\prob{\dim\langle AB\rangle < \dim A\dim
  B}$$ 
that $AB$ is not full-rank is not more than 
$$\sum_{i=1}^\beta\prob{\dim\langle AB_i\rangle <\dim A\dim
  B_{i}\; |\; \dim\langle AB_{i-1}\rangle =\dim A\dim B_{i-1}}$$
so that \eqref{eq:AB'} gives:
\begin{align*}
  \prob{\dim\langle AB\rangle < \dim A\dim B}&\leq
\frac{1}{q-1}\sum_{i=0}^{\beta-1}\left(\frac{1}{q^{m-(t+i+1)\alpha}}-
\frac{1}{q^{m-(t+i)\alpha}}\right)\\ &\leq
\frac{1}{q-1}\left(\frac{1}{q^{m-\alpha(t+\beta)}}-\frac{1}{q^{m-t\alpha}}\right)\leq 
\frac{1}{(q-1)q^{m-\alpha(t+\beta)}}. 
\end{align*}
\end{proof}

We now give the proof of the theorem:

{\em Proof of Theorem \ref{thm:density}.}
To obtain a $T$-decodable syndrome, we must choose $n-k$ elements in a
space $\langle FE\rangle$ for a given space $E$ that contains $T$.
There are $\E(T)$ ways of choosing $E$, and for any given $E$
there are at most
$q^{\dim F\dim E}=q^{dr}$ ways of choosing a syndrome coordinate in
$\langle FE\rangle$. This gives the upper bound on ${\mathcal
  T}(t,r,d,m)$.

We proceed to prove the lower bound. First consider that
Lemma~\ref{lem:AB} proves that, when we randomly and uniformly choose
a subspace $E$ that contains $T$, then with probability at least
$1-1/(q-1)$,
we have:
$$\dim \langle (F_1^{-1}F + F_2^{-1}F)E\rangle = \dim(F_1^{-1}F
+ F_2^{-1}F)\dim E =(2d-1)r$$
by property \eqref{eqa:2dimF-1}. This last fact implies, that
\begin{equation}
  \label{eq:2dr-r}
  \dim (F_1^{-1}\langle FE\rangle + F_2^{-1}\langle FE\rangle) =2dr-r
\end{equation}
since clearly
$$F_1^{-1}\langle FE\rangle + F_2^{-1}\langle FE\rangle =\langle
(F_1^{-1}F + F_2^{-1}F)E\rangle.$$
Now, since we have $E\subset  F_1^{-1}F \cap F_2^{-1}F$, applying the formula
$\dim(A+B) =\dim A +\dim B -\dim A\cap B$ to \eqref{eq:2dr-r} gives us
simultaneously that:
\begin{align*}
  \dim \langle FE\rangle &   = dr\\
   F_1^{-1}F \cap F_2^{-1}F& = E.
\end{align*}
In other words, both conditions (i) and (ii) of $T$-decodability are
satisfied.
We have therefore proved that the proportion of subspaces $E$
containing $T$ that satisfy conditions (i) and (ii) is at least
$(1-1/(q-1))$.
Now let $E$ be a fixed subspace satisfying conditions (i) and (ii).
Among all $(n-k)$-tuples of elements of $\langle FE\rangle$, the
proportion of those $(n-k)$-tuples
that together with $\langle FT\rangle$ generate the whole of $\langle FE\rangle$
is at least
\begin{equation}
  \label{eq:1-1/(q-1)}
  \left(1-\frac 1q\right) \left(1-\frac 1{q^2}\right)\ldots
\left(1-\frac 1{q^i}\right)\ldots \geq 1 - \frac 1{q-1}.
\end{equation}
We have therefore just proved that given a subspace $E$ satisfying conditions
(i) and (ii), there are at least $(1-1/q)(q^{rd})^{n-k}$
$(n-k)$-tuples of $\langle FE\rangle^{n-k}$ satisfying condition
(iii).

To conclude, notice that since a $T$-decodable syndrome entirely
determines the associated subspace $E$, the set of $T$-decodable
syndromes can be partitioned into sets of $(n-k)$-tuples of $\langle
FE\rangle^{n-k}$ satisfying condition (iii) for all $E$ satisfying
conditions (i) and (ii). The two lower bounds on the number of such
$E$ and the number of $T$-decodable syndromes inside a given  $\langle
FE\rangle^{n-k}$ give the global lower bound of the Theorem.
\qed

\section{Differential support attack}\label{app:differential}

We now detail the differential support attack which uses the structure
of the augmented LRPC codes.
 The LRPC code $H$, used to build the signature, is hidden by some
 matrix $S$,$P$ and $R$. As well as any trapdoor cryptosystems, we can
 imagine a specific way to extract the code $H$ from the public key
 $H'=S.(R|H).P$. In this situation, $H$ is defined by $d$ matrices
 $H_1 \dots H_d$ of size $(n-k)\times n$ in $GF(q)$ such that
 $\Sigma_{l=1}^d H_l . F_l = H$. We will provide a specificity of $H'$
 which leads to an exponential extractor of a representation of the
 code $H$ permiting to decode and forge a signature. We give the
 complexity of this extractor and use it as an upper bound for the
 best attack in this cryptosystem.

First, notice that the code $H$ has severals representations. Indeed,
it is constructed using $H_1 \dots H_d$ and $F_1 \dots F_d$. Here we
want to choose a canonical representation to simplify the proof. For
that purpose, we search for the $n\times n$ matrix $P'$ instead of $P$
in $GF(q)$ such that $H'=S(R|Id.F_1\dots Id.F_d).P'$ with $Id$ the
identity matrix. We can find such a matrix because the parameters are
choosen such as $d(n-k)=n$ and the matrix $H_l$ have rank $(n-k)$,
with $1 \le l \le d$. We can also choose, without loss of generality,
a homogeneous form for $F_1,\dots ,F_d$ where $F_1=1$. This can be
deduced by swapping the matrices $S$ and $S.\frac{1}{F_1}$. Below we
try to extract a code $H$ of the form  $(Id|Id.F_2|\dots|Id.F_d)$.

In this paragraph we describe the vector space in $GF(q)$ generated by the element in a line of $H'$. We set $(S_{i,j})_{1 \le i,j \le n-k}$ for the coefficients of $S$ and  $(R_{i,j})_{(1 \le i \le n-k)(1 \le j \le t)}$ for the coefficients of $R$. The coefficient $(i,j)$ of the matrix $S.(R|H)$ can be expressed by :  
  \begin{itemize}
	\item  $\Sigma_{p=1}^{n-k}S_{i,p}R_{p,j}$, if $1 \le j \le t$
	\item $S_{i,j-t}$, if $t \le j \le t+n-k$ 
	\item $\dots$
	\item$S_{i,j-k-t}F_d$, if $k+t \le j \le n+t$ 
\end{itemize}
Each element of the row $i$ of the matrix $S.(R|H)$ belongs to the
$GF(q)$-vector space $V_i = \langle S_{i,1}F_1,\dots ,S_{i,n-k}F_1,\dots ,S_{i,n-k}F_d, R_1, \dots ,R_t\rangle$ with $R_1,\dots ,R_t$ some coefficients depending on $S$ and $R$. 
Eventually, the multiplication by the matrix $P$ on the right does not change that each element of the $i$-th row of $H'$ belongs to the vector space $V_i$.\\    
It is {\it a priori} difficult to retrieve an element  $F_l$, $1 < l
\le d$, from one of the $V_i$. On the other hand, we can verify that
an element $\alpha$ is a $F_l$ by computing $V_i \cap V_i . \alpha
^{-1}$. If $\alpha$ is one of the $F_l$, the intersection will be
$\langle S_{i,1},\dots ,S_{i,n}\rangle$  for all $1 \le i \le n-k$. Then we can
retrieve $\langle F_1, \dots, F_d\rangle$ with the intersection
$\cup_{p=1}^{n-k}V_i.\frac{1}{S_{i,p}}$. As long as $d$ is not a large
number, it is not difficult to extract the whole structure from that. 

A simple way to find a $F_l$ is to test any possibilities in $GF(q^m)$
with the intersection described before. We will see next a more
efficient method which uses the repetition of the element $F_l$, $1 <
l \le d$, in the rows of the LRPC code $H$.

The search for one of the $(F_l)_{1 < l \le d}$ is based on the search
for an element $x$ of one of the $V_i$ such that $x =
\Sigma_{j=1}^{n-k}\lambda_j S_{i,j}.F_l$, with $\lambda_j$ in $GF(q)$
and $l \neq 1$. We know that there exists another element
$x'=\Sigma_{j=1}^{n-k}\lambda_j S_{i,j}.F_1$ in $V_i$. Hence, there
exists a combination in the vector space $V_i . \frac{1}{x}$ equal to
$\frac{1}{F_l}$. This corresponds to a combination of elements in the
row $i$ of $H . \frac{1}{x}$ which would be equal to
$\frac{1}{F_l}$. If we find a combination $c\in GF(q)^{n+t}$ such that
$H . \frac{1}{x} . c^T=(v_1,\dots,v_{n-k})^T$ and $v_i=\frac{1}{F_l}$
for a particular $i$, then we have that $v_i=\frac{1}{F_l}$ for all $1
\le i \le n-k$. Indeed the different $V_i$ are built in a same way
from $H'$ and if a combination $c$ gives $L_ic^T =
\Sigma_{j=1}^{n-k}\lambda_j S_{i,j}$ with $L_i$ the row $i$ of $H'$
then we have  by construction $L_pc^T = \Sigma_{j=1}^{n-k}\lambda_j
S_{p,j}$ with the same $\lambda_j$. Finally, to retrieve this
combination $c$, we can look for $c$ such that
$(L_1.\frac{1}{x_1}-L_2.\frac{1}{x_2})c^T=0$ where $x_1$ is picked as
random in $V_1$ and $x_2$ is generated with the same algorithm that
$x_1$ but using $V_2$.

We obtain a complexity based on the probability of randomly finding a useful
element $x$: $q^{(n-k)(d-1)+t}$.

This point of view shows that even a specific attack on the hidden code $H$ which uses all its particularities will not succeed with well chosen parameters since the complexity
remains exponential. The invertible matrix $P$ in $GF(q)$ seems to
sufficiently mix the vector spaces $V_i$ to make it difficult to
recover  a vector $x$ which could allow one to extract the structure of~$H$.

\section{Proof of Theorem~\ref{thm:forgery}}\label{app:leakage}

Recall that a signature of a message $M$ is a pair $(x',y')$ where $y'$ is
a hashed value of the message $M$ and
$y'=H'x'^T$ and $\rank(x')=r$. If $\A$ is an algorithm that leads to a
forgery with the use of $N$ authentic signatures, then
the algorithm $\A'$ consists of a simulated
version of $\A$ where authentic signatures $(x',y')$ are replaced by
couples $(x'',y'')$ where
$x''$ is randomly and uniformly chosen among vectors of rank-weight
  $e$, and $y''=H'x''^T$ is claimed to be the hashed value of the
  message $M$ output by a random oracle. In the random oracle model,
  the algorithm $\A'$ must behave exactly as algorithm $\A$ and give
  the same output whenever $(x'',y'')$ is statistically
  indistinguishable from $(x',y')$.

We now compare the statistics of $(x',y')$ and $(x'',y'')$.
We have $H'=A(R|H)P$ and
since the transformation:
\begin{align}
  (x',y') &\mapsto (x^a=Px',y^a=A^{-1}y')\label{eq:xaya}\\
  (x'',y'') &\mapsto (x^s=Px'',y^s=A^{-1}y'')\label{eq:xsys}
\end{align}
is one-to-one, comparing the statistics of $(x',y')$ and $(x'',y'')$
amounts to comparing the distributions of $(x^a,y^a)$ (authentic) and 
$(x^s,y^s)$ (simulated).

Now $(x^a,y^a)$ is obtained in the following way: the signer
chooses a subspace $T$ of $GF(q^m)$ together with a random vector
$\tau\in T^t$ of rank $t$ and is given
a vector $u$ which is uniformly distributed in the syndrome space
$GF(q^m)^{n-k}$ and is equal to $A^{-1}h(M)-R\tau^T$. 
Precisely, the signer chooses a random vector $\tau$ of rank $t$, and
sets $T$ to be the subspace generated by its coordinates.
The signer then
proceeds to try to decode $u$, meaning it looks for a subspace $E$
of $GF(q^m)$ that contains $T$ and such that all coordinates of $u$
fall into $\langle FE\rangle$, where $F$ is the space generated by the elements of
the LRPC matrix $H$. He succeeds exactly when the syndrome vector $u$
is
$T$-decodable in the sense of Definition~\ref{def:Tdecodable}: when
this doesn't occur, the decoder aborts.

When the syndrome $u$ is $T$-decodable, the decoder proceeds
to solve the equation
\begin{equation}
  \label{eq:HxH=u}
  Hx_H^T = u
\end{equation}
and then sets 
$$x^a = (\tau,x_H)$$
to create the couple $(x^a,y^a=(R|H)(x^a)^T)$ in \eqref{eq:xaya}.
Recall from Remark 1 in Section~\ref{sec:UsingLRPC}
that the matrix $H$ has been chosen so
that equation~\eqref{eq:HxH=u} (equivalently
equation~\eqref{eq:mappingEn}) always has a unique solution for every
$T$-decodable $u$.


We may therefore speak about $T$-decodable {\em couples} $(x_H,u)$, where
$u$ uniquely determines $x_H$ and $x_H$ uniquely determines $u$.
Now, to re-cap, the authentic signer starts with uniformely random $u$, and
whenever $u$ turns out to be non $T$-decodable, 
then we declare a decoding failure and start the whole
process again generating another $\tau$, another random space $T$ and 
another $u$ by
another call to the random oracle $h$ (meaning a counter appended to
the message $M$ is incremented before applying the random hash again).
This keeps happening until we hit a $T$-decodable $u$.
 We see therefore that when it does hit a
$T$-decodable $u$, the couple 
  $$(x_H,u)$$
is uniformly distributed among all $T$-decodable couples.

We now turn to the action of the simulator: what the simulator does is
he tries to generate a uniform $T$-decodable couple $(x_H,u)$ through $x_H$
rather than through $u$ like the signer.

Specifically, the simulator starts with a random subspace $E$ of $GF(q^m)$
of dimension $r$ and an $x''$ with coordinates independently and
uniformly drawn from $E$. Since the transformation \eqref{eq:xsys} $x''\mapsto
x^s=Px''$ is rank-preserving, the simulator is implicitely creating a
uniform vector $x^s$ of $E^n$. Write
$$x^s=(\tau, x_H).$$
With overwhelming probability (at least $1-1/q^{r-t}$), the vector $\tau\in E^t$ is of maximum
rank-weight $t$, since its coordinates are independently and uniformly
chosen in $E$. The vector $\tau$ generates the required random space
$T$. Let $u=Hx_H^T$: note that by construction, all
its coordinates must be in $\langle FE\rangle$. Consider the conditions
(i),(ii),(iii) of Definition~\ref{def:Tdecodable} for $u$ to be $T$-decodable. 

Remember that the first two conditions (i) and (ii) are properties {\em only of the subspace $E$}. When they are not satisfied, no choice of $x_H$ can
yield a $T$-decodable couple $(x_H,u)$. 
When conditions (i) and (ii) are satisfied we have that,
since the mapping \eqref{eq:HxH=u} $x_H\mapsto u$ is invertible,
the vector
$u$ is a uniform random vector in $\langle FE\rangle^{n-k}$. Since $n-k + \dim \langle FT\rangle =
\dim \langle FE\rangle$, the probability that condition (iii) is not satisfied is governed by the probability than a random
vector falls into a given subspace of $\langle FE\rangle$ of co-dimension $1$ and is
of the order of $1/q$: it is also at most $1/(q-1)$ according to
computation
\eqref{eq:1-1/(q-1)}. We also see that  the number of $x_H$
that satisfies condition (iii) is independent of the space $E$ and
is always the same for all $E$ that satisfy conditions (i) and (ii).
This last fact implies that when 
\begin{itemize}
\item A random uniform subspace $E$ is chosen among all possible
  subspaces $E$ of dimension~$r$,
\item a random $x_H$ is chosen in $E^n$,
\end{itemize}
then either $(x_H,u)$ is not $T$-decodable, or $(x_H,u)$ is 
$T$-decodable and is uniformly distributed among all 
$T$-decodable couples.

The simulator has no oracle to tell him when he has produced a
non $T$-decodable couple, he can only hope this doesn't occur.
As long as he produces $T$-decodable couples $(x_H,u)$, then they
are distributed (uniformly) exactly as those that are produced by
the authentic signer and are undistinguishable from them. 
If we call $\pi$ the probability that he produces a non-decodable
$u$, then he can reasonably expect to produce a list of
approximately $N=1/\pi$ signatures $(x'',y'')$ that are undistiguishable
from genuine signatures in the random oracle model.

Consider now the probability $\pi$ that the simulator produces a
non-decodable $u$. It is at most the sum of the probabilities that $E$
does not satisfy (i) and (ii) and the probability $\leq 1/(q-1)$ that
(iii) is not satisfied. The probability that $E$ does not satisfy (i)
and (ii) is at most the probability that the
product space $\langle (F_1^{-1}F + F_2^{-1})E\rangle$ does not have
maximal dimension $(2d-1)r$, as argued in the proof of
Theorem~\ref{thm:density},
and is at most $1/(q-1)$. We obtain therefore $\pi\leq 1/(q-1) +
1/(q-1) = 2/(q-1)\approx 2/q$ which concludes the proof.

\end{document}